\documentclass{llncs} %\usepackage[spanish,english]{babel}
\usepackage{epsfig}
\usepackage{url}

\usepackage{amssymb,amsfonts,amsmath,enumerate,latexsym}
\usepackage{mathptmx} % selects Times Roman as basic font
\usepackage{helvet} % selects Helvetica as sans-serif font
\usepackage{courier} % selects Courier as typewriter font
\usepackage{type1cm} % activate if the above 3 fonts are
\usepackage{makeidx} % allows index generation
\usepackage{graphicx} % standard LaTeX graphics tool
\usepackage{multicol} % used for the two-column index
\usepackage[bottom]{footmisc}% places footnotes at page bottom

\usepackage{tikz}
\usetikzlibrary{snakes,shapes} \tikzstyle{hyb}=[rectangle,draw,minimum
size=3mm] \tikzstyle{tre}=[circle,draw] \tikzstyle{tri}=[regular
polygon,regular polygon sides=5,draw,minimum size=2mm]
\tikzstyle{indef}=[regular polygon,regular polygon sides=5,draw]
\newcommand{\etq}[1]{ \draw (#1) node {\scriptsize $#1$};}
\tikzstyle{path}=[snake=coil,segment aspect=0,->]

 \renewcommand{\geq}{\geqslant}

 \makeindex

\begin{document}
% Si hace su comunicaci\'on en ingl\'es utilice:
% \selectlanguage{english} en lugar de:
%\selectlanguage{english}

\title{%
  A metric for galled networks \thanks{%
    This work has been partially supported by the Spanish Government
    and FEDER funds, through project MTM2009-07165 and
    TIN2008-04487-E/TIN.  } }
% Utilice \titlerunning{T\'{\i}tulo corto} para una versi\'on
% abreviada del t\'{\i}tulo de su contibuci\'on si el original resulta
% demasiado largo

\author{ Gabriel Cardona%\inst{1}%
  \and Merc\`e Llabr\'es%\inst{1}%
  \and Francesc Rossell\'o%\inst{1}%
}
% Utilice \authorrunning{T\'{\i}tulo corto} para una versi\'on
% abreviada del t\'{\i}tulo de su contribuci\'on si el original
% resulta demasiado largo.

\institute{ Department of Mathematics \& Computer Science, University
  of Balearic Islands,\\Crta. Valldemossa, km. 7,5. 07122 -- Palma de
  Mallorca,
  Spain. \url{{gabriel.cardona,merce.llabres,cesc.rossello}@uib.es}
%
  % \and Nombre y direcci\'on de la Instituci\'on
  % \url{{2oNombre,3erNombre}@domicilio.correo}%
}
% Utilice el paquete "url.sty" para evitar problemas con los
% car\'acteres especiales utilizados en su direcci\'on web o de correo
\maketitle

\begin{abstract}
  Galled networks, directed acyclic graphs that model evolutionary
  histories with reticulation cycles containing only tree nodes, have
  become very popular due to both their biological significance and
  the existence of polynomial time algorithms for their
  reconstruction. In this paper we prove that Nakhleh's $m$ measure is
  a metric for this class of phylogenetic networks and hence it can be
  safely used to evaluate galled network reconstruction methods.
\end{abstract}

\setcounter{footnote}{0}
\section{Introduction}

Phylogenetic networks have been studied over the last years as a
richer model of the evolutionary history of sets of organisms than
phylogenetic trees, because they take into account not only mutation
events but also reticulation events, like recombinations,
hybridizations, and lateral gene transfers.  Technically, it is
accomplished by modifying the concept of phylogenetic tree in order to
allow the existence of nodes with in-degree greater than one.  As a
consequence, much progress has been made to find practical algorithms
for reconstructing a phylogenetic network from a set of sequences or
other types of evolutive information.  Since different reconstruction
methods applied to the same sequences, or a single method applied to
different sequences, may yield different phylogenetic networks for a
given set of species, a sound measure to compare phylogenetic networks
becomes necessary~\cite{moret.ea:2004}.  The comparison of
phylogenetic networks is also needed in the assessment of phylogenetic
reconstruction methods \cite{moret:alenex05}, and it will be required
to perform queries on future databases of phylogenetic
networks~\cite{page:05}.

Several distances for the comparison of phylogenetic networks have
been proposed so far in the literature, including generalizations to
networks of the Robinson-Foulds distance for trees, like the
tripartitions distance \cite{moret.ea:2004} and the $\mu$-distance
\cite{CLRV:bioinfo,CRV:mu}, and different types of nodal distances
\cite{CLRV:comp2,CLRV:nodnet}. All polynomial time computable
distances for phylogenetic networks introduced up to now do not
separate arbitrary phylogenetic networks, that is, zero distance does
not imply in general isomorphism. Of course, this is consistent with
the equivalence between the isomorphism problems for phylogenetic
networks and for graphs, and the general belief that the latter lies
in NP$-$P.  Therefore one has to study for which interesting classes
of phylogenetic networks these distances are metrics in the precise
mathematical sense of the term. The interest of the classes under
study may stem from their biological significance, or from the
existence of reconstruction algorithms.

This work contributes to this line of research. We prove that a
distance introduced recently by Nakhleh \cite{Nak:m} separates
semibinary galled networks (roughly speaking, networks where every
node of in-degree greater than one has in-degree exactly two and every
reticulation cycle has only one hybrid node; see the next section for
the exact definition, \cite{HK07,HR+08} for a discussion of the
biological meaning of this condition, and \cite{Nak.ea:06,Nak.ea:05}
for reconstruction algorithms). In this way, this distance turns out
to be the only non-trivial metric available so far on this class of
networks that is computable in polynomial-time.

\section{Preliminaries}

Given a set $S$ of \emph{labels}, a \emph{$S$-DAG} is a directed
acyclic graph with its leaves bijectively labelled by $S$.  In a
$S$-DAG, we shall always identify without any further reference every
leaf with its label.

Let $N=(V,E)$ be a $S$-DAG. A node is a \emph{leaf} if it has
out-degree $0$ and \emph{internal} otherwise, a \emph{root} if it has
in-degree $0$, of \emph{tree} type if its in-degree is $\le 1$, and of
\emph{hybrid} type if its in-degree is $> 1$.  $N$ is \emph{rooted}
when it has a single root.  A node $v$ is a \emph{child} of another
node $u$ (and hence $u$ is a \emph{parent} of $v$) if $(u,v)\in
E$. Two nodes with a parent in common are \emph{sibling} of each
other. A node $v$ is a \emph{descendant} of a node $u$ when there
exists a path from $u$ to $v$: we shall also say in this case that $u$
is an \emph{ancestor} of $v$.  The \emph{height} $h(v)$ of a node $v$
is the largest length of a path from $v$ to a leaf.

A \emph{phylogenetic network} on a set $S$ of \emph{taxa} is a rooted
$S$-DAG such that no tree node has out-degree $1$ and every hybrid
node has out-degree $1$. A \emph{phylogenetic tree} is a phylogenetic
network without hybrid nodes.   A \emph{reticulation cycle} in a
phylogenetic network is a pair of 
  internally disjoint paths from a tree node (its \emph{source}) to a
  hybrid node (its \emph{target}).

The underlying biological motivation for these definitions is that
tree nodes model species (either extant, the leaves, or non-extant,
the internal tree nodes), while hybrid nodes model reticulation
events. The parents of a hybrid node represent the species involved in
this event and its single child represents the resulting species (if
it is a tree node) or a new reticulation event where this resulting
species gets involved into without yielding any other descendant (if
the child is a hybrid node). The tree children of a tree node
represent direct descendants through mutation. The absence of
out-degree 1 tree nodes in phylogenetic network means that every
non-extant species has at least two different direct descendants. This
is a very common restriction in any definition of phylogeny, since
species with only one child cannot be reconstructed from biological
data.

Many restrictions have been added to this definition. Let us introduce
now some of them. For more information on these restrictions,
including their biological or technical motivation, see the references
accompanying them.

\begin{itemize}
\item A phylogenetic network is \emph{semibinary} if every hybrid node
  has in-degree $2$ \cite{CLRV:bioinfo}, and \emph{binary} if it is
  semibinary and every internal tree node has out-degree $2$.

\item A phylogenetic network is a \emph{galled network}, when every
  non-target node in every reticulation is a tree
  node~\cite{HK07,HR+08}.

\end{itemize}

Two hybridization networks $N,N'$ are \emph{isomorphic}, in symbols
$N\cong N'$, when they are isomorphic as directed graphs and the
isomorphism sends each leaf of $N$ to the leaf with the same label in
$N'$.

\section{On Nakhleh's distance $m$}

Let us recall the distance $m$ introduced by Nakhleh in
\cite{Nak:m}, %nakhleh:phd},
in the version described in \cite{CLRV:m}.  Let $N=(V,E)$ be a
phylogenetic network on a set $S$ of taxa. For every node $v\in V$,
its \emph{nested label} $\lambda_N(v)$ (or simply $\lambda(v)$ when
there is no risk of confusion) is defined by recurrence as
follows:
\begin{itemize}
\item If $v$ is the leaf labelled $i$, then $\lambda_N(v)=\{i\}$.

\item If $v$ is internal and all its children $v_1,\ldots,v_k$ have
  been already labelled, then $\lambda_N(v)$ is the multiset
  $\{\lambda_N(v_1),\ldots,\lambda_N(v_k)\}$ of their labels.
\end{itemize}
The absence of cycles in $N$ entails that this labelling is
well-defined.

Notice that the nested label of a node is, in general, a nested
multiset (a multiset of multisets of multisets of\ldots), hence its
name.  Moreover, the height of a node $u$ is the highest level of
nesting of a leaf in $\lambda(u)$ minus 1.

Now, it is easy to prove from the nested label definition, the
following result.
\begin{lemma}
  Let $N=(V,E)$ be a phylogenetic network on a set $S$ of taxa.
  \begin{itemize}
  \item If $(u,v)\in E$, then $\lambda(v)\in \lambda(u)$;
  \item If there is a path from $u$ to $v$, then there exists a set of
    nodes $u_1,...,u_k$ such that $u_1=u$, $u_k=v$ and
    $\lambda(u_i)\in \lambda(u_{i+1})$ for every $i=1,...,k-1$.\qed
  \end{itemize}
\end{lemma}

The \emph{nested labels representation} of $N$ is the multiset
$$
\lambda(N)=\{\lambda_N(v)\mid v\in V\},
$$
where each nested label appears with multiplicity the number of nodes
having it as nested label.  \emph{Nakhleh's distance} $m$ between a
pair of phylogenetic networks $N,N'$ on a same set $S$ of taxa is then
$$
m(N,N')=|\lambda(N)\bigtriangleup \lambda(N')|,
$$
where the symmetric difference and the cardinal refer to multisets.

This distance trivially satisfies all axioms of metrics except, at
most, the separation axiom, and thus this is the key property that has
to be checked on some class of networks in order to guarantee that $m$
is a metric on it.  So far, this distance $m$ is known to be a metric
for reduced networks \cite{Nak:m}, tree-child networks \cite{CLRV:m},
and semibinary tree-sibling time consistent networks \cite{CLRV:m}
(always on any fixed set of labels $S$). It is not a metric for
arbitrary tree-sibling time consistent networks \cite{CLRV:m}. And, we
will prove here, that it is a metric for semibinary galled networks,
which implies that it is also a metric for galled trees and 1-nested
networks.

\section{The distance $m$ for galled networks}
\label{sec:distance-m-galled}

In this section we prove that the distance $m$ defined above separates
galled networks up to isomorphism.

First of all, notice that if a galled network $N=(V,E)$ has no pair of
different nodes with the same nested label, then for every pair of
nodes $u,v\in V$, we have that $(u,v)\in E$ iff $\lambda_N(v)\in
\lambda_N(u)$. Indeed, on the one hand, the very definition of nested
label entails that if $(u,v)\in E$, then $\lambda_N(v)\in
\lambda_N(u)$; and conversely, if $\lambda_N(v)\in \lambda_N(u)$, then
$u$ has a child $v'$ such that $\lambda_N(v')=\lambda_N(v)$, and by
the injectivity of nested labels, it must happen that $v=v'$.

This clearly implies that a galled network without any pair of
different nodes with the same nested label can be reconstructed, up to
isomorphisms, from its nested labels representation, and hence that
non-isomorphic galled networks without any pair of different nodes
with the same nested label always have different nested label
representations. Therefore, it remains to prove the separation axiom
of Nakhleh's distance for galled networks with some pair of different
nodes with the same nested label.

The general result will be proved by algebraic induction on the number
of pairs of different nodes with the same nested label. To this end,
we introduce a pair of reduction procedures that decrease the number
of pairs of different nodes with the same nested label in a semibinary
galled network.  Each of these reductions, when applied to a galled
network with $n$ leaves and with at least one pair of
different nodes with the same nested label, produces a galled network
with $n$ leaves and one pair less of different nodes with the same
nested label. Moreover, given any galled network with more than one
leaf and with at least one pair of different nodes with the same
nested label, it is always possible to apply to it some of these
reductions.

\smallskip
\begin{enumerate}
\item[($\mathbf{R}$)] Let $N$ be a galled network, let $u\neq v$ be a
  pair of sibling nodes such that $ \lambda(u)= \lambda(v)$ and assume
  that $u$ and $v$ have the same children which are the hybrid nodes
  $h_1,...,h_k$. The \emph{$R_{u;v;h_1,...,h_k}$ reduction} of $N$ is
  the network $R_{u;v;h_1,...,h_k}(N)$ obtained by removing the nodes
  $u,v,h_1,...,h_k$, together with their incoming arcs, and adding an
  arc from the parent of $u$ and $v$ to each child of $h_1,...,h_k$;
  cf.~Fig.~\ref{fig:R-red}.\footnote{In graphical representations of
    hybridization networks, we shall represent hybrid nodes by
    squares, tree nodes by circles, and indeterminate (that is, that
    can be of tree or hybrid type) nodes by pentagons.}

  \smallskip

\item[($\mathbf{T}$)] Let $N$ be a galled network, let $u\neq v$ be a
  pair of no sibling nodes such that $ \lambda(u)= \lambda(v)$ and
  assume that $u$ and $v$ have the same children which are the hybrid
  nodes $h_1,...,h_k$. Let $x,y$ be the parents of $u,v$ respectively,
  and notice that these nodes must be of tree type, since otherwise
  $N$ would contain a reticulation cycle with hybrid internal
  nodes. The \emph{$T_{u;v;h_1,...,h_k}$ reduction} of $N$ is the
  network $T_{u;v;h_1,...,h_k}(N)$ obtained by removing the nodes
  $u,v,h_1,...,h_k$, together with their incoming arcs, and adding a
  hybrid node $h$ with a tree child $w$ and arcs from $x$ and $y$ to
  $h$, from $h$ to $w$, and from $w$ to each child of $h_1,...,h_k$;
  cf.~Fig.~\ref{fig:T-red}.
\end{enumerate}

Notice that in both cases, the resulting network is a galled network
since, in the first case, we simply remove hybrid nodes, and in the
second one, we simply replace $k$ hybridization cycles by only one,
without adding any hybrid intermediate node. Also, in both cases the
number of pairs of different nodes with the same nested label
decreases in a unit. Last, we remark that in any case, the nodes
$w_1,\dots,w_k$ in the resulting network have disjoint sets of
descendants. Indeed, if some different nodes $w_i$ and $w_j$ share a
descendant $y$, then there exists a common hybrid descendant $h$ and a
reticulation cycle having $h$ as its target and $x$ as its
source. This cycle would induce a cycle in the original network that
would contain hybrid nodes, hence yielding a contradiction.

\begin{figure}
  % \centering
  \begin{tikzpicture}[scale=0.3]
  \end{tikzpicture}
\end{figure}

\begin{figure}[htb]
  \centering
  \begin{tikzpicture}[thick,>=stealth,scale=0.3]
    \draw (0,2) node[tri] (xx) {}; \draw (0,0) node[tre] (x) {};
    \etq{x} \draw (-4,-2) node[tre] (u) {};\etq{u} \draw (4,-2)
    node[tre] (v) {};\etq{v} \draw (7,-2) node (c) {$\cdots$}; \draw
    (-4,-4) node[hyb] (h_1) {};\etq{h_1} \draw (0,-4) node (d)
    {$\cdots$}; \draw (4,-4) node[hyb] (h_k) {};\etq{h_k} \draw
    (-4,-6) node[tre] (w_1) {}; \etq{w_1} \draw (0,-6) node
    {$\cdots$}; \draw (4,-6) node[tre] (w_k) {}; \etq{w_k} \draw
    (-6,-8) node[tre] (x_1) {}; \draw (-4,-8) node {$\cdots$}; \draw
    (-2,-8) node[tre] (x_2) {}; \draw (0,-8) node {$\cdots$}; \draw
    (2,-8) node[tre] (y_1) {}; \draw (4,-8) node {$\cdots$}; \draw
    (6,-8) node[tre] (y_2) {};

    \draw[->](xx) to (x); \draw[->](x) to (u); \draw[->](x) to (v);
    \draw[->](x) to (c); \draw[->](u) to (h_1); \draw[->](u) to (h_k);
    \draw[->](u) to (d); \draw[->](v) to (h_1); \draw[->](v) to (h_k);
    \draw[->](v) to (d); \draw[->](h_1) to (w_1); \draw[->](h_k) to
    (w_k); \draw[->](w_1) to (x_1); \draw[->](w_1) to (x_2);
    \draw[->](w_k) to (y_1); \draw[->](w_k) to (y_2);

\end{tikzpicture}
\qquad
\begin{tikzpicture}[thick,>=stealth,scale=0.3]
  \draw(0,-2) node{\ }; \draw(0,0) node {$\Rightarrow$}; \draw(0,2)
  node{\ };
\end{tikzpicture}
\qquad
\begin{tikzpicture}[thick,>=stealth,scale=0.3]
  \draw (0,2) node[tri] (xx) {}; \draw (0,0) node[tri] (x) {}; \etq{x}
  \draw (-4,-3) node[tre] (w_1) {}; \etq{w_1} \draw (0,-3) node (c)
  {$\cdots$}; \draw (4,-3) node[tre] (w_k) {}; \etq{w_k} \draw (6,-2)
  node (d) {$\cdots$}; \draw (-6,-5) node[tre] (x_1) {}; \draw (-4,-5)
  node {$\cdots$}; \draw (-2,-5) node[tre] (x_2) {}; \draw (0,-5) node
  {$\cdots$}; \draw (2,-5) node[tre] (y_1) {}; \draw (4,-5) node
  {$\cdots$}; \draw (6,-5) node[tre] (y_2) {};

  \draw[->](xx) to (x); \draw[->](x) to (w_1); \draw[->](x) to (w_k);
  \draw[->](x) to (c); \draw[->](x) to (d); \draw[->](w_1) to (x_1);
  \draw[->](w_1) to (x_2); \draw[->](w_k) to (y_1); \draw[->](w_k) to
  (y_2);
\end{tikzpicture}

\caption{\label{fig:R-red} The $R_{u,v;h_1,\dots,h_k}$ reduction.}
\end{figure}

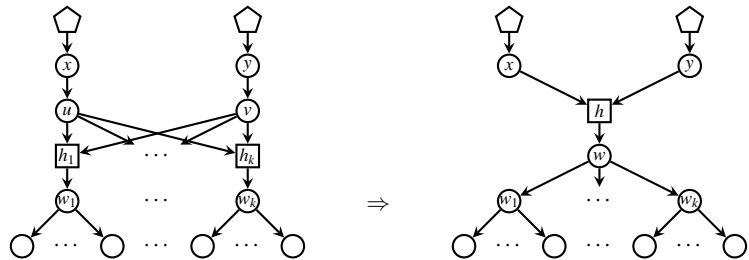
\begin{figure}[htb]
  \centering
  \begin{tikzpicture}[thick,>=stealth,scale=0.3]
    \draw (-4,2) node[tri] (xx) {}; \draw (4,2) node[tri] (yy) {};
    \draw (-4,0) node[tre] (x) {}; \etq{x} \draw (4,0) node[tre] (y)
    {}; \etq{y} \draw (-4,-2) node[tre] (u) {};\etq{u} \draw (4,-2)
    node[tre] (v) {};\etq{v}
    % \draw (7,-2) node (c) {$\cdots$};
    \draw (-4,-4) node[hyb] (h_1) {};\etq{h_1} \draw (0,-4) node (d)
    {$\cdots$}; \draw (4,-4) node[hyb] (h_k) {};\etq{h_k} \draw
    (-4,-6) node[tre] (w_1) {}; \etq{w_1} \draw (0,-6) node
    {$\cdots$}; \draw (4,-6) node[tre] (w_k) {}; \etq{w_k} \draw
    (-6,-8) node[tre] (x_1) {}; \draw (-4,-8) node {$\cdots$}; \draw
    (-2,-8) node[tre] (x_2) {}; \draw (0,-8) node {$\cdots$}; \draw
    (2,-8) node[tre] (y_1) {}; \draw (4,-8) node {$\cdots$}; \draw
    (6,-8) node[tre] (y_2) {};

    \draw[->](xx) to (x); \draw[->](yy) to (y); \draw[->](x) to (u);
    \draw[->](y) to (v); \draw[->](u) to (h_1); \draw[->](u) to (h_k);
    \draw[->](u) to (d); \draw[->](v) to (h_1); \draw[->](v) to (h_k);
    \draw[->](v) to (d); \draw[->](h_1) to (w_1); \draw[->](h_k) to
    (w_k); \draw[->](w_1) to (x_1); \draw[->](w_1) to (x_2);
    \draw[->](w_k) to (y_1); \draw[->](w_k) to (y_2);

\end{tikzpicture}
\qquad
\begin{tikzpicture}[thick,>=stealth,scale=0.3]
  \draw(0,-2) node{\ }; \draw(0,0) node {$\Rightarrow$}; \draw(0,2)
  node{\ };
\end{tikzpicture}
\qquad
\begin{tikzpicture}[thick,>=stealth,scale=0.3]
  \draw (-4,2) node[tri] (xx) {}; \draw (4,2) node[tri] (yy) {}; \draw
  (-4,0) node[tre] (x) {}; \etq{x} \draw (4,0) node[tre] (y) {};
  \etq{y} \draw (0,-2) node[hyb] (h) {}; \etq{h} \draw (0,-4)
  node[tre] (w) {}; \etq{w} \draw (-4,-6) node[tre] (w_1) {};
  \etq{w_1} \draw (0,-6) node (c) {$\cdots$}; \draw (4,-6) node[tre]
  (w_k) {}; \etq{w_k}
  % \draw (6,-2) node (d) {$\cdots$};
  \draw (-6,-8) node[tre] (x_1) {}; \draw (-4,-8) node {$\cdots$};
  \draw (-2,-8) node[tre] (x_2) {}; \draw (0,-8) node {$\cdots$};
  \draw (2,-8) node[tre] (y_1) {}; \draw (4,-8) node {$\cdots$}; \draw
  (6,-8) node[tre] (y_2) {};

  \draw[->](xx) to (x); \draw[->](yy) to (y); \draw[->](x) to (h);
  \draw[->](y) to (h); \draw[->](h) to (w); \draw[->](w) to (w_1);
  \draw[->](w) to (c); \draw[->](w) to (w_k); \draw[->](w_1) to (x_1);
  \draw[->](w_1) to (x_2); \draw[->](w_k) to (y_1); \draw[->](w_k) to
  (y_2);
\end{tikzpicture}

\caption{\label{fig:T-red} The $T_{u,v;h_1,\dots,h_k}$ reduction.}
\end{figure}

Now we have the following basic applicability result.

\begin{proposition}\label{thm:red}
  Let $N$ be a galled network with a pair of different nodes with the
  same nested label.
 Then, at least one $R$ or $T$ reduction can be applied to $N$, and
  the result is a galled network.
\end{proposition}

\begin{proof}
  Let $N$ be a galled network with a pair of nodes $u\neq v$ such that
  $\lambda_N(u)=\lambda_N(v)$. Without any loss of generality, we
  assume that $v$ is a node of smallest height among those nodes with
  the same nested label as some other node. By definition, $v$ cannot
  be a leaf, because the only node with nested label $\{i\}$, with
  $i\in S$, is the leaf labelled $i$. Therefore $v$ is internal: let
  $v_1,\ldots,v_k$ ($k\geq 1$) be its children, so that
  $\lambda_N(v)=\{\lambda_N(v_1),\ldots, \lambda_N(v_k)\}$.  Since
  $\lambda_N(u)=\lambda_N(v)$, $u$ has $k$ children, say
  $u_1,\ldots,u_k$, and they are such that
  $\lambda_N(v_i)=\lambda_N(u_i)$ for every $i=1,\ldots,k$. Then,
  since $v_1,\ldots,v_k$ have smaller height than $v$ and by
  assumption $v$ is a node of smallest height among those nodes with
  the same nested label as some other node, we deduce that $v_i=u_i$
  for every $i=1,\ldots,k$. Therefore, $v_1,\ldots,v_k$ are hybrid,
  and their only parents (by the semibinarity condition) are $u$ and
  $v$.  Hence, we can apply the $R$ reduction when $u$ and $v$ are
  sibling and the $T$ reduction when they have different parents.

  The fact that the result of the application of a $R$ or a $T$
  reduction to $N$ is again a galled network has been discussed in the
  definition of the reductions.\qed
\end{proof}

We shall call the inverses of the $R$ and $T$ reductions,
respectively, the $R^{-1}$ and $T^{-1}$\emph{expansions}, and we shall
denote them by $R_{u;v,h_1,...,h_k}^{-1}$ and
$T^{-1}_{u;v,h_1,...,h_k}$.  More specifically, for every galled
network $N$:
\begin{itemize}

\item if $N$ contains a tree node $x$ with tree children nodes
  $w_1,...,w_k$ such that they do not have any descendant node in
  common, then the $R_{u;v,h_1,...,h_k}^{-1}$ expansion can be applied
  to $N$, and $R_{u;v,h_1,...,h_k}^{-1}(N)$ is obtained by removing
  the arcs from $x$ to its children $w_1,\dots,w_k$, adding two tree
  nodes $u,v$, and $k$ hybrid nodes $h_1,...,h_k$, together with arcs
  from $x$ to $u$ and $v$, from $u$ and $v$ to every added hybrid
  node, and from $h_i$ to $w_i$ for every $i=1,...,k$;

\item if $N$ contains a hybrid node $h$, whose only child $w$ has $k$
  children tree nodes $w_1,...,w_k$ such that they do not have any
  descendant node in common, then the $T_{u;v,h_1,...,h_k}^{-1}$
  expansion can be applied to $N$, and $T_{u;v,h_1,...,h_k}^{-1}(N)$
  is obtained by removing the hybrid node $h$ and its child $w$
  together with their incoming and outgoing arcs and adding two tree
  nodes $u,v$ and $k$ hybrid nodes $h_1,...,h_k$ together with arcs
  from one parent of $h$ to $u$ and from the other parent of $h$ to
  $v$, from $u$ and $v$ to every added hybrid node, and from $h_i$ to
  $w_i$ for every $i=1,...,k$.
\end{itemize}

From these descriptions, since $w_1,...,w_k$ do not have any
descendant node in common, we easily see that the result of a $R^{-1}$
or $T^{-1}$ expansion applied to a galled network is always a galled
network.

%%%
%%%
The following result is easily deduced from the explicit descriptions
of the reductions and expansions.

\begin{lemma}
  % \label{lem:iso-red}
  Let $N$ and $N'$ be two galled networks.  If $N\cong N'$, then the
  result of applying to both $N$ and $N'$ the same $R^{-1}$ expansion
  (respectively, $T^{-1}$ expansion) are again two isomorphic galled
  networks.

  Moreover, if we apply a $R$ or $T$ reduction to a galled network
  $N$, then we can apply to the resulting network the corresponding
  inverse $R^{-1}$ or $T^{-1}$ expansion and the result is a galled
  network isomorphic to $N$.  \qed
\end{lemma}

So, by \cite[Lem.~6]{CLRV:galled}, and since galled networks without
any pair of different nodes with the same nested label always have
different nested label representations, to prove that the Nakhleh's
distance $m$ separates semibinary galled networks, it is enough to
prove that the possibility of applying a reduction to a semibinary
galled network $N$ can be decided from $\lambda(N)$, and that the
nested label representation of the result of the application of a
reduction to a semibinary galled network $N$ depends only on
$\lambda(N)$ and the reduction. These two facts are given by the the
following lemmas.

\begin{lemma}
  % \label{lem:qb-H1}
  Let $N$ be a galled network on a set $S$.

  \begin{enumerate}[(1)]
  \item If a reduction $R_{u;v,h_1,...,h_k}$ can be applied to $N$,
    then the nodes $u,v$ involved in the reduction satisfy the
    following property:
    $\lambda(u)=\lambda(v)=\{\{\lambda(w_1)\},...,\{\lambda(w_k)\}\}$
    and there is a node $x$ such that
    $\{\lambda(u),\lambda(v)\}\subseteq\lambda(x)$.
  \item Conversely, if two nodes $u,v$ satisfy the property above, and
    have minimal height among those that safisfy it, then a reduction
    $R_{u;v,h_1,...,h_k}$ can be applied to $N$.

  \item If $R_{u;v,h_1,...,h_k}$ can be applied to $N$, let $N'$ be
    the resulting network. Then $\lambda(N')$ can be computed from
    $\lambda(N)$ as follows:
    \begin{itemize}
    \item If $A\in\lambda(N)$ does not contain any $\lambda(w_i)$ at
      any level of nesting, then $A\in\lambda(N')$;
    \item If $A\in\lambda(N)$ is equal to some $\lambda(w_i)$, then
      $A\in\lambda(N')$;
    \item If $A\in\lambda(N)$ contains at some level of nesting the
      element $\{\{A_1\},\dots,\{A_k\}\}$ (with multiplicity $2$),
      where $A_i=\lambda(w_i)$, then replace it by the $k$ elements
      $A_1,\dots,A_k$ (with multiplicity $1$) to get
      $A'\in\lambda(N')$.
    \end{itemize}
  \end{enumerate}
\end{lemma}
\begin{proof}
  If a $R$ reduction can be applied to a galled network $N$, then it
  is clear that there exists two sibling nodes $u$ and $v$ and hence a
  common parent $x$ such that $\lambda(u)=\lambda(v)$,
  $\lambda(x)\supseteq\{\lambda(u),\lambda(v)\}$. Since $u$ and $v$
  have the same children nodes, these nodes, say $h_1,\dots,h_k$, must
  hybrid.  For each $i=1,\dots,k$, let $w_i$ be the single child of
  $h_i$.  Then it is clear that
  $\lambda(u)=\lambda(v)=\{\{\lambda(w_1)\},...,\{\lambda(w_k)\}\}$.

  Conversely, assume that %there exists two nodes $u\neq v$ such that
  $\lambda(u)=\lambda(v)=\{\{\lambda(w_1)\},...,\{\lambda(w_k)\}\}$. From
  the nested label definition and the minimality assumption on the
  height, this implies that $u$ and $v$ have $k$ children nodes which
  are hybrid nodes. Moreover, since there is a node $x$ such that
  $\lambda(x)\supseteq\{\lambda(u),\lambda(v)\}$, we can conclude that
  $u$ and $v$ are sibling nodes and then, we can apply a $R$ reduction
  to $N$.

  Now, if $R_{u;v,h_1,...,h_k}$ can be applied to $N$, then
  $N'=R_{u;v,h_1,...,h_k}(N)$ is the galled network obtained by
  removing the nodes $u,v,h_1,...,h_k$, together with their incoming
  arcs, and adding an arc from the parent of $u$ and $v$ to each child
  of $h_1,...,h_k$. Thus, the nested label of $w_i$ and the nested
  label of all those nodes being descendant nodes of $w_i$ for every
  $i=1,..k$ remains the same as in $N$. In the same way, the nested
  label of all those nodes not being ancestors of $w_i$ for every
  $i=1,...,k$ remains the same as in $N$, and then, they are in the
  nested label representation of $N'$. Finally, the nested label of
  the ancestors of $w_i$ for every $i=1,..,k$ must be relabeled since
  we have delete two intermediate nodes in every path from the
  ancestor to $w_i$. This implies, that we delete two levels of
  nesting, one tree node and $k$ hybrid nodes, and then, we must
  replace $\{\{\{A_1\},...,\{A_k\}\},\{\{A_1\},...,\{A_k\}\},\dots\}$
  by
  $\{A_1,...,A_k,\dots\}$.\qed 
\end{proof}

\begin{lemma}
  Let $N$ be a galled network on a set $S$.

  \begin{enumerate}[(1)]
  \item If a reduction $T_{u;v,h_1,...,h_k}$ can be applied to $N$,
    then the nodes $u,v$ involved in the reduction satisfy the
    following property:
    $\lambda(u)=\lambda(v)=\{\{\lambda(w_1)\},...,\{\lambda(w_k)\}\}$
    and there is not any node $x$ such that
    $\{\lambda(u),\lambda(v)\}\subseteq\lambda(x)$.
  \item Conversely, if two nodes $u,v$ satisfy the property above, and
    have minimal height among those that safisfy it, then a reduction
    $T_{u;v,h_1,...,h_k}$ can be applied to $N$.

  \item If $T_{u;v,h_1,...,h_k}$ can be applied to $N$, let $N'$ be
    the resulting network. Then $\lambda(N')$ can be computed from
    $\lambda(N)$ as follows:
    \begin{itemize}
    \item If $A\in\lambda(N)$ does not contain any $\lambda(w_i)$ at
      any level of nesting, then $A\in\lambda(N')$;
    \item If $A\in\lambda(N)$ is equal to some $\lambda(w_i)$, then
      $A\in\lambda(N')$;
    \item If $A\in\lambda(N)$ contains at some level of nesting the
      element $\{\{A_1\},\dots,\{A_k\}\}$, where $A_i=\lambda(w_i)$,
      then replace it by the elements $\{\{A_1,\dots,A_k\}\}$ (with
      the same multiplicity) to get $A'\in\lambda(N')$.
    \item Include also $\{A_1,\dots,A_k\}$ and $\{\{A_1,\dots,A_k\}\}$
      in $\lambda(N')$.
    \end{itemize}
  \end{enumerate}
\end{lemma}

\begin{proof}
  The proof of this lemma goes the same way as the previous one,
  taking now into account how the nested labels are modified.\qed
\end{proof}

As a result, we get the desired result.

\begin{theorem}
  The distance $m$ defined above is a metric on the space of galled
  networks on a fixed set of labels.\qed
\end{theorem}

% \section*{Acknowledgements}
%   This work has been partially supported by the Spanish Government and
%   FEDER funds, through project MTM2009-07165 and TIN2008-04487-E/TIN.
%\end{acknowledgement}
% \verb+ pendent revisar bibliografia+

\end{document}